\documentclass[11pt]{article}

\usepackage{amsfonts,amsmath,amsthm,times,fullpage,amssymb}

\usepackage[hidelinks]{hyperref}

\usepackage{enumitem}

\usepackage[dvipsnames]{xcolor}

\usepackage[noend]{algpseudocode} 
\usepackage[nothing]{algorithm}  
\usepackage{caption}
\usepackage[normalem]{ulem}

\newtheorem{theorem}{Theorem} 
\newtheorem{lemma}{Lemma}
\newtheorem{definition}{Definition} 
 
\newtheorem{remark}{Remark}

\newtheorem*{harmonic}{Harmonic Walks}
\newtheorem*{regeneration}{Regeneration}
\newtheorem*{gLLL}{General LLL}

\newtheorem*{cd}{Causality Digraph}
\newtheorem*{causality}{Potential Causality}

\newcommand{\Indep}{\mathrm{Ind}}
\newcommand{\List}{\mathrm{List}}  
\newcommand{\Span}{\mathcal{S}}

\newcommand{\param}{\psi}

\begin{document}

\title{
A Local Lemma for Focused Stochastic Algorithms
\footnote{This paper is based on results that appeared in preliminary form in~\cite{Harmonic} and were improved, among other contributions, in~\cite{Kolmofocs}.}}

\author{
Dimitris Achlioptas
\thanks{Research supported by NSF grant CCF-1514128.}\\ 
Department of Computer Science\\ 
University of California Santa Cruz\\
\and 
Fotis Iliopoulos
\thanks{Research supported by NSF grant CCF-1514434.} \\ 
Department of Electrical Engineering and Computer Science\\ 
University of California Berkeley
\and 
Vladimir Kolmogorov
\thanks{Research supported by the European Research Council under the 
Seventh Framework Programme (FP7/2007-2013), ERC grant agreement no 616160.}\\ 
Institute of Science and Technology Austria
}

\date{\empty}

\maketitle

\begin{abstract}
We develop a framework for the rigorous analysis of focused stochastic local search algorithms. These are algorithms that search a state space by repeatedly selecting some constraint that is violated in the current state and moving to a random nearby state that addresses the violation, while hopefully not introducing many new ones. An important class of focused local search algorithms with provable performance guarantees has recently arisen from algorithmizations of the Lov\'{a}sz Local Lemma (LLL), a non-constructive tool for proving the existence of satisfying states by introducing a background measure on the state space. While powerful, the state transitions of algorithms in this class must be, in a precise sense, perfectly compatible with the background measure. In many applications this is a very restrictive requirement and one needs to step outside the class.  Here we introduce the notion of  \emph{measure distortion} and develop a framework for analyzing arbitrary focused stochastic local search algorithms, recovering LLL algorithmizations as the special case of no distortion. Our framework takes as input an arbitrary such algorithm and an arbitrary probability measure and shows how to use the measure as a yardstick of algorithmic progress, even for algorithms designed independently of the measure.
\end{abstract}

\newpage

\section{Introduction}\label{sec:intro}

Let $\Omega$ be a large, finite set of objects and let $F = \{f_1, f_2, \ldots, f_m \}$ be a collection of subsets of $\Omega$. We will refer to each $f_i \in F$ as a \emph{flaw} to express that all objects in $f_i$ have negative feature $i \in [m]$. For example, for a CNF formula on $n$ variables with clauses $c_1,c_2,\ldots,c_m$, for each clause $c_i$ we can define $f_i \subseteq \{0,1\}^n$ to comprise the truth assignments that violate $c_i$. Following linguistic rather than mathematical convention we will say that flaw $f$ is present in object $\sigma$ if $f \ni \sigma$ and that $\sigma \in \Omega$ is \emph{flawless} if no flaw is present in $\sigma$. 

The Lov\'{a}sz Local Lemma is a non-constructive tool for proving the \emph{existence} of flawless objects by introducing a probability measure $\mu$ on $\Omega$ and bounding from below the probability of simultaneously avoiding all (``bad") events corresponding to the flaws in $F$. (Below and throughout we assume that products devoid of factors evaluate to 1.)

\begin{gLLL}
Given events $A_1,\ldots,A_m$, for each $i \in [m]$, let the set $D(i) \subseteq [m] \setminus \{i\}$ be such that if $S \subseteq [m] \setminus (D(i) \cup \{i\})$, then $\mu(A_i \mid \cap_{j \in S} \overline{A_j}) = \mu(A_i)$. If there exist $\{\psi_i\}_{i=1}^m>0$ 
such that for all $i \in [m]$,
\begin{equation}\label{eq:LLL}
\frac{\mu(A_i)}{\psi_i}  \sum_{ S \subseteq \{i \} \cup D(i) }  \prod_{j \in S} \psi_j \le 1 \enspace , 
\end{equation}
then the probability that none of $A_1,\ldots,A_m$ occurs is at least $\prod_{i=1}^m 1/(1+\psi_i) > 0$. 
\end{gLLL}

Erd\H{os} and Spencer~\cite{LopsTrav} noted that independence in the LLL can be replaced by \emph{positive correlation}, yielding the stronger Lopsided LLL. The  difference is that each set $D(i)$ is replaced by a set $L(i) \subseteq [m] \setminus \{ i \}$ such that if $S \subseteq [m] \setminus (L(i) \cup \{i\})$, then $\mu(A_i \mid \cap_{j \in S} \overline{A_j}) \le \mu(A_i)$, i.e., ``=" is replaced by ``$\le$". Also, condition~\eqref{eq:LLL} is more well-known as $\mu(A_i) \le x_i \prod_{j \in D(i)} (1-x_j)$, where $x_i = \psi_i/(1+\psi_i)$. As we will see, formulation~\eqref{eq:LLL} better facilitates the statement of refinements of the condition. Specifically, considering graphical properties of the graphs on $[m]$ induced by the relationships $D(\cdot)$ and $L(\cdot)$, one  can show more permissive conditions, such as the cluster expansion~\cite{bissacot2011improvement}, the lefthanded~\cite{PegdenLLLL}, and Shearer's condition~\cite{Shearer}. 

Moser~\cite{moser}, later joined by Tardos~\cite{MT}, in groundbreaking work showed that a simple algorithm can be used to make the general LLL constructive when $\mu$ is a product measure. Specifically, in the \emph{variable setting} of~\cite{MT}, each event $A_i$ is determined by a set of variables $\mathrm{vbl}(A_i)$ so that $ j \in D(i)$ iff $\mathrm{vbl}(A_i) \cap \mathrm{vbl}(A_j) \ne \emptyset$. Moser and Tardos proved that in the variable setting, if~\eqref{eq:LLL} holds, then repeatedly selecting \emph{any} occurring event $A_i$ and resampling every variable in $\mathrm{vbl}(A_i)$ independently according to $\mu$, leads to a flawless object after a linear expected  number of resamplings. Pegden~\cite{PegdenIndepen} proved that this remains true under the weakening of~\eqref{eq:LLL} to the cluster expansion condition of Bissacott et al.~\cite{bissacot2011improvement}. Finally, Kolipaka and Szegedy~\cite{szege_meet} proved that the resampling algorithm actually works even under Shearer's tight condition~\cite{Shearer}. In an orthogonal development, Harris and Srinivasan in~\cite{SrinivasanPerm} were the first to make the LLL constructive outside the variable setting, giving an algorithmic LLL for the uniform measure on permutations.

Moser's original analysis of the resampling algorithm in the context of satisfiability~\cite{moser} inspired a parallel line of works that formed the so-called \emph{entropy compression} method, e.g.,~\cite{dujmovic2016nonrepetitive, Highrepetitive, acyclic}. In these works, the set of objects $\Omega$ typically does not have product structure, there is no measure $\mu$, and no general condition for algorithmic convergence. Instead, the fact that the algorithm under consideration must reach a flawless object (and, thus, terminate) is established by proving that the entropy of its trajectory grows more slowly than the rate at which it consumes randomness. The rate comparison is done in each case via a problem-specific counting argument.

A common feature of the resampling algorithm of~\cite{MT}, the swapping algorithm of~\cite{SrinivasanPerm}, and all algorithms analyzed by entropy compression, is that they are instances of focused stochastic local search. The general idea in stochastic local search is that $\Omega$ is equipped with a neighborhood structure, so that the search for flawless objects starts at some (flawed) object (state) and moves stochastically from state to state along the neighborhood structure. Restricting the search so that every transition away from a state $\sigma$ must target one of the flaws present in $\sigma$ is known as \emph{focusing} the search~\cite{papafocus}. The first effort to give a convergence condition for focused stochastic local search algorithms with arbitrary transition probabilities, i.e., not mandated by a background measure $\mu$, was the flaws/actions framework of~\cite{AIJACM}. As our work also uses this framework, below we recall some of the relevant definitions. Note that the existence of flawless objects is not presumed in any of these analyses. Instead, the idea is to establish the existence of flawless objects by proving that some focused stochastic local search algorithm (quickly) converges to one.

For $\sigma \in \Omega$, let $U(\sigma)$ denote the set of indices of the flaws present in $\sigma$, i.e., $U(\sigma) = \{ i \in [m]: f_i \ni \sigma \}$. For every $i \in U(\sigma) $,  let $A(i,\sigma) \ne \{ \sigma \}$ be a non-empty subset of $\Omega$. The elements of $A(i,\sigma)$ are called \emph{actions} and we consider the multi-digraph $D$ on $\Omega$ that has an arc $\sigma \xrightarrow{i} \tau $ for every $\tau \in A(i,\sigma)  $. We will consider walks on $D$ which start at a state $\sigma_1$ selected according to some probability distribution $\theta$, and which at each non-sink vertex $\sigma$ first select a flaw $f_i \ni \sigma$ as a function of the trajectory so far (focus), and then select as the next state $\tau \in A(i,\sigma) $ with probability $\rho_i(\sigma, \tau)$ . Whenever flaw $f_i \ni \sigma$ is selected we will say that flaw $f_i$ was \emph{addressed}. This will not necessarily mean that $f_i$ will be eliminated, i.e., potentially $A(i,\sigma) \cap f_i \ne \emptyset$. The multidigraph $D$ should  be thought of as implicitly defined by the algorithm we wish to analyze each time, not as explicitly constructed. Also, when we refer to running ``time", we will refer to the number of steps on $D$, without concern for exactly how long it takes to perform a single step, i.e., to identify a flaw present and select from its actions.

In this language, \cite{AIJACM} gave a sufficient condition for algorithmic convergence when:
\begin{enumerate}[label=(\alph*)]
\item\label{cond:atomicity}
$D$ is \emph{atomic}, i.e., for every $\tau \in \Omega$ and every $i\in [m]$ there exists at most one arc $\sigma \xrightarrow{i} \tau$.
\item\label{cond:uniformity}
$\rho$ assigns \emph{equal} probability to every action in $A(i,\sigma)$, for every $\sigma \in \Omega$ and $i \in U(\sigma)$.  
\end{enumerate}
By analyzing algorithms satisfying conditions~\ref{cond:atomicity} and~\ref{cond:uniformity}, several results that had been proved by custom versions of the LLL, and thus fell outside the algorithmization framework of~\cite{MT}, were made constructive and improved in~\cite{AIJACM}. At the same time, the convergence condition of~\cite{AIJACM} makes it possible to recover most results of the entropic method by generic arguments (sometimes with a small parameter loss). Finally, it is worth pointing out that even though the framework of~\cite{AIJACM} does not reference a background probability measure $\mu$, it captures a large fraction of the applications of general LLL. This is because when $\mu$ is uniform and bad events correspond to partial assignments, a very common scenario, the state transitions of the resampling algorithm of Moser and Tardos satisfy both conditions~\ref{cond:atomicity} and~\ref{cond:uniformity}. Overall, though, the convergence condition of~\cite{AIJACM} was incomparable with those of the LLL algorithmizations preceding it.

The long line of work on LLL algorithmizations that started with the groundbreaking work of Moser, culminated with the work of Harvey and Vondr\'{a}k~\cite{HV}. They showed that the Lopsided LLL can be made constructive even under the most permissive condition (Shearer's), whenever one can construct efficient \emph{resampling oracles}. Resampling oracles elegantly capture the common core of all LLL algorithmizations, namely that the state transitions, $(D,\rho)$, are perfectly compatible with the background measure $\mu$. Below we give the part of the definition of resampling oracles that exactly expresses this notion of compatibility, which we dub (measure) \emph{regeneration}.
\begin{regeneration}[Harvey-Vondr\'{a}k~\cite{HV}]
\label{def:loc_regen}
Say that $(D,\rho)$ \emph{regenerate $\mu$ at flaw $f_i$} if for every $\tau \in \Omega$,
\begin{equation}\label{eq:oracle}
\frac{1}{\mu(f_i)}\sum_{\sigma \in f_i} \mu(\sigma) \rho_{i}(\sigma,\tau) = \mu(\tau) \enspace .
\end{equation}
\end{regeneration}

Observe that the l.h.s.\ of~\eqref{eq:oracle} is the probability of reaching state $\tau$ after first sampling a state $\sigma \in f_i$ according to $\mu$ and then addressing $f_i$ at $\sigma$. The requirement that this probability equals $\mu(\tau)$ for every $\tau \in \Omega$ means that $(D,\rho)$ must be such that in every state $\sigma \in f_i$ the distribution on actions for addressing $f_i$ \emph{perfectly} removes the conditional $f_i \ni \sigma$. Of course, a trivial way to satisfy this requirement is to sample a new state $\sigma'$ according to $\mu$ in each step (assuming $\mu$ is efficiently sampleable). Doing this, though, foregoes any notion of iterative progress towards a goal, as the set of flaws present in $\sigma'$ are completely unrelated to those in $\sigma$. Instead, one would like to respect~\eqref{eq:oracle} while limiting the flaws introduced in $\sigma'$. To that end, we can consider the projection of the action digraph $D$ capturing which flaws may be introduced (caused) when we address each flaw. It is important to note that, below, potential causality is independent of flaw choice and that the causality digraph has an arc $i \to j$ if there exists \emph{even one} transition aimed at addressing $f_i$ that causes $f_j$ to appear in the new state. Naturally, the sparser this causality digraph, the better. 

\begin{causality}
For an arc $\sigma \xrightarrow{i}  \tau$ in $D$ and a flaw $f_j$ present in $\tau$ we say that $f_i$ causes $f_j$ if $f_i = f_j$ or $f_j \not\ni \sigma$. We say that $f_i$ \emph{potentially causes} $f_j$ if $D$ contains \emph{at least one} arc wherein $f_i$ causes $f_j$.
\end{causality}
\begin{cd}
The digraph $C=C(\Omega,F,D)$ on $[m]$ where $i \rightarrow j$ iff $f_i$ potentially causes $f_j$ is called the  causality digraph. The \emph{neighborhood} of a flaw $f_i$  is $\Gamma(i) =\{j : i \to  j \text{  exists in $C$}\}$.
\end{cd}
Harvey and Vondr\'{a}k~\cite{HV} proved that for essentially every lopsidependency digraph $L$ of interest, there exist resampling oracles whose causality digraph is (a subgraph of) $L$. We should emphasize, though, that there is no guarantee that these promised resampling oracles can be implemented efficiently, so as to yield an LLL algorithmization (and, naturally, in the absence of efficiency considerations the LLL is already ``algorithmic" by exhaustive search). Indeed, as we discuss below, there are settings in which the existence of efficient resampling oracles seems unlikely. That said, in~\cite{HV} Harvey and Vondr\'{a}k demonstrated the existence of efficient resampling oracles for a plethora  of LLL applications in the variable setting, the permutation setting, and several other settings.

Perhaps the simplest demonstration of the restrictiveness of resampling oracles comes from one of the oldest and most vexing concerns about the LLL (see the survey of Szegedy~\cite{mario_survey}). Namely, the inability\footnote{Naturally, whenever the set of flawless objects $\Omega^*$ is non-empty, the uniform measure on $\Omega^*$ demonstrates the existence of flawless objects. So, in a trivial sense, there is nothing that can not be established by the LLL. But, of course, anyone in possession of a description of $\Omega^*$ allowing the construction of a measure on it, does not need the LLL. Indeed, the whole point of the LLL is that it offers incredibly rich conclusions, e.g., $\Omega^* \neq \emptyset$, from extremely meager ingredients, e.g., the uniform measure on $\Omega$.} of the LLL to establish that a graph with maximum degree $\Delta$ can be colored with $ q = \Delta + 1$ colors. For example, if $\mu$ is the uniform measure on all $q^n$ colorings with $q$ colors, then every time a vertex $v$ is recolored, its color must be chosen uniformly among all colors, something that induces a requirement of $q > \mathrm{e}\Delta$ colors. If, instead, one could chose only among colors that do not currently appear in $v$'s neighborhood, then for all $q \ge \Delta+1$, the causality digraph is empty and rapid termination follows trivially. But it seems very hard to describe a probability measure $\mu$ and resampling oracles for it that respect the empty causality graph.

To recap, there are two ``schools of thought." In the first, one starts from the central object of the LLL, the measure $\mu$ on $\Omega$, and tries to design an algorithm that moves from one state to another in a manner that perfectly respects the measure. In the other, there is no measure on $\Omega$ at all and both the transitions and their probabilities can be, a priori, arbitrary. In this work, we bring these two schools of thought  together by introducing the notion of \emph{measure distortion}, showing, in particular, that the first school corresponds to the special case of no distortion. The main point of our work, though, is to demonstrate that the generality afforded by allowing  measure distortion has tangible benefits. Specifically, in complex applications, the requirement that every resampling must perfectly remove the conditional of the resampled bad event can be impossible to meet by short travel within $\Omega$, i.e., by ``local search". This is because small, but non-vanishing, correlations can travel arbitrarily far in the structure. Allowing measure distortion removes the requirement of perfect deconditioning, with any correlation seepage (distortion) is accounted for, via a local analysis. This makes it possible to design natural, local algorithms and prove rigorous mathematical statements about their convergence in the presence of long-range correlations. 
 
Concretely, we extend the flaws/actions framework of~\cite{AIJACM} to allow arbitrary action digraphs $D$, arbitrary transition probabilities $\rho$, and the incorporation of arbitrary background measures $\mu$, allowing us to connect the flaws/actions framework to the Lov\'{a}sz Local Lemma. Our work highlights the role of the measure $\mu$ in gauging how efficiently the algorithm rids the state from flaws, i.e., as a \emph{gauge of progress}, by pointing out the trade-off between distortion and the sparsity of the causality graph. The end result is a theorem that subsumes both the results of~\cite{AIJACM} and the algorithmization of the Lopsided LLL~\cite{HV} via resampling oracles, establishing a uniform method for designing and analyzing focused stochastic local search algorithms. Additionally, our work makes progress on elucidating the role of flaw choice in stochastic local search, and establishes several structural facts about resampling oracles.

\section{Statement of Results}\label{sec:results}

We develop tools for analyzing \emph{focused stochastic local search} algorithms. Specifically, we establish a sequence of increasingly general conditions under which such algorithms find flawless objects quickly, presented as Theorems~\ref{asymmetric},\ref{olala}, and \ref{lem:master}.  For the important special case of atomic action digraphs we identify structural properties of resampling oracles, presented as Theorem~\ref{atomic_oracles}. For the same setting we also derive a sharp analysis for the probability of any trajectory, elucidating the role of flaw choice, presented as Theorem~\ref{tight}. 

Theorems~\ref{asymmetric}--\ref{lem:master} differ in the sophistication of the flaw-choice mechanism they can accommodate. While in works such as~\cite{MT} on the variable setting and~\cite{SrinivasanPerm} on permutations, the setting was sufficiently symmetric that flaw choice could be arbitrary, in more complex applications more sophisticated flaw-choice is necessary. For example, to establish our results on Acyclic Edge Coloring we must use our recursive algorithm (Theorem~\ref{olala}), as the simple Markov walk (Theorem~\ref{asymmetric}), let alone arbitrary flaw choice, will not work. 

To demonstrate the flexibility of our framework, we derive a bound for Acyclic Edge Coloring of graphs with bounded degeneracy, a class including all graphs of bounded treewidth, presented as Theorem~\ref{Aecaki} in Section~\ref{AECARA}. To derive the result we rely heavily on the actions \emph{not} forming resampling oracles with respect to the measure used. Unlike other recent algorithmic work on the problem~\cite{acyclic, kirousis}, our result is established without ideas/computations ``customized" to the problem, but as a direct application of Theorem~\ref{olala}, highlighting its capacity to incorporate both global conditions, such as degeneracy, and sophisticated flaw-choice mechanisms, in this case a recursive procedure. We also show how to derive effortlessly an upper bound of $4.182 (\Delta-1)$ for Acyclic Edge Coloring of general graphs, which comes close to the hard-won bound of $4 (\Delta-1)$ of Esperet and Parreau~\cite{acyclic} via a custom analysis. Finally, we note that  Iliopoulos~\cite{LLLWTL} recently showed how our main theorem can be used to analyze the algorithm of Molloy~\cite{molloy2017list} for coloring triangle-graph graphs of degree up to the ``shattering threshold" for random graphs~\cite{mitsaras_barriers}.

\subsection{Setup}\label{sec:setup}
Recall that we consider algorithms which at each flawed state $\sigma$ select some flaw $f_i \ni \sigma$ to address and then select the next state $\tau \in A(i,\sigma)$ with probability $\rho_i(\sigma,\tau)$. As one may expect the flaw choice mechanism does have a bearing on the running time of such algorithms and we discuss this point in Section~\ref{sec:fc}. Our results focus on conditions for rapid termination that do not require sophisticated flaw choice (but can be used in conjunction which such choice). 

To measure a walk's capacity to rid the state of flaws we introduce a measure $\mu$ on $\Omega$, as in the LLL. Without loss of generality, and to avoid certain trivialities, we assume that $\mu(\sigma)>0$ for all $\sigma \in \Omega$.  The choice of $\mu$ is entirely ours and can be  oblivious, e.g., $\mu(\cdot) = |\Omega|^{-1}$. While  $\mu$ will typically assigns only exponentially small probability to flawless objects, it will allow us to prove that the walk reaches a flawless object in polynomial time with high probability. 

To do this we define a ``charge" $\gamma_i = \gamma_i(D,\theta,\rho,\mu)$ for each flaw $f_i \in F$ that captures the \emph{compatibility} between the actions of the algorithm for addressing flaw $f_i$ and the measure $\mu$.  Specifically, just as for regeneration, we consider the probability, $\nu_i(\tau)$, of ending up in state $\tau$ after (i) sampling a state $\sigma \in f_i$ according to $\mu$, and then (ii) addressing $f_i$ at $\sigma$. But instead of requiring that $\nu_i(\tau)$ \emph{equals} $\mu(\tau)$, as in resampling oracles, we allow $\nu_i(\tau)$ to be free and simply measure
\begin{align}
d_i =  \max_{\tau \in \Omega } \frac{\nu_i (\tau) }{\mu (\tau) }  \ge 1 \enspace ,
\end{align}
i.e., the greatest inflation of a state probability incurred by addressing $f_i$ (relative to its probability under $\mu$, and averaged over the initiating state $\sigma \in f_i$ according to $\mu$). The charge $\gamma_i$ of flaw $f_i$ is then defined as 
\begin{eqnarray}\label{eq:charges}
\gamma_i & := & d_i \cdot  \mu(f_i) \\
& = &  \max_{\tau \in \Omega}  \frac{1}{ \mu(\tau) } \sum_{ \sigma \in f_i } \mu( \sigma )  \rho_i(\sigma, \tau )   \enspace . 
\end{eqnarray}
To gain some intuition for $\gamma_i$ observe that if $\mu$ is uniform and $D$ is atomic, then $\gamma_i$ is simply the greatest transition probability $\rho_i$ on any arc originating in $f_i$.

To state our results we need a last definition regarding the distribution $\theta$ of the starting state. 

\begin{definition}
The \emph{span} of a probability distribution $\theta: \Omega \rightarrow [0,1]$, denoted by $\Span(\theta)$, is the set of flaw indices that may be present in a state selected according to $\theta$, i.e., $ \Span(\theta) = \bigcup_{ \sigma \in \Omega: \theta(\sigma) > 0 } U(\sigma)$.
\end{definition}

\subsection{A Simple Markov Chain}

Our first result concerns the simplest case where in each flawed state $\sigma$, the algorithm addresses the greatest flaw present in $\sigma$, according to an arbitrary but fixed permutation of the flaws. Recall that $\mu$ is the measure on $\Omega$ used to measure progress, $\gamma_i$ is the charge of flaw $f_i$ according to $\mu$, and $\theta$ is the starting state distribution. 
\begin{theorem}\label{asymmetric}
If there exist positive real numbers $\{\psi_i\}_{i \in [m]}$ such that for every $i \in [m]$,
\begin{align}
\zeta_i := 
\frac{\gamma_i}{\psi_i}
 \sum_{ S \subseteq   \Gamma(i) }  \prod_{j \in S} \psi_j   < 1\enspace , \label{eq:mc2}
\end{align}
then for every permutation $\pi$, the walk reaches a sink within $(T_0+s)/\delta$ steps with probability at least $1-2^{-s}$, where $\delta = 1 - \max_{i \in [m]} \zeta_i 
> 0$, and
\[
T_0 	 = 
\log_2 \left( \max_{ \sigma \in \Omega } \frac{\theta(\sigma)  }{ \mu(\sigma) }\right)  + \log_2 \left( \sum_{S \subseteq \Span(\theta) } \prod_{j \in S} \psi_j   \right) 
=  \log_2 \left( \max_{ \sigma \in \Omega } \frac{\theta(\sigma)  }{ \mu(\sigma) }\right)  +  \sum_{j \in \Span(\theta) } \log_2(1+\psi_j)
\enspace .
\]
\end{theorem}

Theorem~\ref{asymmetric} has two features worth discussing, shared by all our results.\medskip

\noindent {\bf Arbitrary starting state.} Since $\theta$ can be arbitrary, any foothold on $\Omega$ suffices to apply the theorem. Note also that $T_0$ captures the trade-off between starting at a fixed state vs.\ starting at a state sampled from $\mu$. In the latter case, i.e., when $\theta = \mu$, the first term in $T_0$ vanishes, but the second term grows to reflect the uncertainty of the set of flaws present in $\sigma_1$.\smallskip

\noindent {\bf Arbitrary number of flaws.} The running time depends only on the span $|\Span(\theta)|$, not the total number of flaws $|F|$. This has an implication  analogous to the result of Hauepler, Saha, and Srinivasan~\cite{haeupler_lll} on core events: even when $|F|$ is super-polynomial in the problem's encoding length, it may still be possible to get a polynomial-time algorithm. For example, this can be done by proving that in every state only polynomially many flaws may be present, or by finding a specific state $\sigma_1$ such that $|U(\sigma_1)|$ is small.

\subsection{A Non-Markovian Algorithm}

Our next results concerns the common setting where the neighbors of each flaw in the causality graph span several arcs between them.  We improve Theorem~\ref{asymmetric} in such settings by employing a \emph{recursive} algorithm. That is, an algorithm where the flaw choice at each step depends on the \emph{entire} trajectory up to that point, not just the current state, so that the resulting walk on $\Omega$ is non-Markovian. It is for this reason that we required a non-empty set of actions for every flaw present in a state, and why the definition of the causality digraph does not involve flaw choice. The improvement is that rather than summing over all subsets of $\Gamma(i)$ as in~\eqref{eq:mc2}, we now only sum over \emph{independent} such subsets, where $f_i,f_j$ are dependent if $i \rightarrow j$ \emph{and} $j \rightarrow i$. This improvement is similar to the cluster expansion improvement of Bissacot et al.~\cite{bissacot2011improvement} of the general LLL. As a matter of fact, Theorem~\ref{olala} implies the algorithmic aspects of~\cite{bissacot2011improvement} (see~\cite{PegdenIndepen} and ~\cite{HV}). 

Further, the use of a recursive algorithm makes it possible to  ``shift responsibility" between flaws, so that gains from the aforementioned restriction of the sum can be realized by purposeful flaw ordering. For a permutation $\pi$ of $F$, let $I_{\pi}(S)$ denote the index of the greatest flaw in any $S \subseteq F$ according to $\pi$. For a fixed action digraph $D$ with causality digraph $C$, the recursive algorithm \emph{takes as input} any digraph $R \supseteq C$, i.e., any supergraph of $C$, and is the non-Markovian random walk on $\Omega$ that occurs by invoking procedure {\sc Eliminate}. Observe that if in line~\ref{a_key_diff} we do not intersect $U(\sigma)$ with $\Gamma_R(i)$ the recursion is trivialized, recovering the simple walk of Theorem~\ref{asymmetric}. Its convergence condition, Theorem~\ref{olala}, involves sums over the independent sets of $R$, generalizing the discussion above (as one can always take $R=C$).   
\begin{algorithm}\caption*{{\bf Recursive Walk}}
\begin{algorithmic}[1]\label{Recursive}
\Procedure{Eliminate}{}
\State $\sigma \leftarrow \theta(\cdot)$ \Comment{Sample $\sigma$ from $\theta$}
\While {$U(\sigma) \neq \emptyset$}	
	\State {\sc Address} ($I_{\pi}(U(\sigma)),\sigma$) 
\EndWhile	
\State	\Return $\sigma$ 
\EndProcedure{}{}
\Procedure{Address}{$i,\sigma$}
\State $\sigma \leftarrow$ $\tau \in A(i,\sigma)$ with probability $\rho_i(\sigma,\tau)$ 
\While {$B = U(\sigma) \cap \Gamma_{R}(i) \neq \emptyset$}   \label{a_key_diff}	\Comment{Note $\,\cap \Gamma_{R}(f_i)$} \label{code:critical}
	\State {\sc{Address}}($I_{\pi}(B),\sigma$) 				\label{code:act}
\EndWhile  
\EndProcedure  
\end{algorithmic}
\end{algorithm}

The reason for allowing the addition of arcs in $R$ relative to $C$ is that while adding, say, arcs $i\to j$ and $j \to i$ may make the sums corresponding to $f_i$ and $f_j$ greater, if flaw $f_k$ is such that $\{i, j\} \subseteq \Gamma(k)$, then the sum for flaw $f_k$ may become smaller, since $f_i,f_j$ are now dependent. As a result, without modifying the algorithm, such arc addition can help establish a sufficient condition for rapid convergence to a flawless object, e.g., in our application on Acyclic Edge Coloring in Section~\ref{AECARA}. An analogous phenomenon is also true in the improvement of Bissacot et al.~\cite{bissacot2011improvement}, i.e., denser dependency graphs may yield better analysis.

\begin{definition}\label{defn:G}
For a digraph $R$ on $[m]$, let $G=G(R)=([m],E)$ be the \emph{undirected} graph where $\{i,j\} \in E$ iff both $i \rightarrow j$ and $j \rightarrow i$ exist in $R$. For $S \subseteq F$, let $\Indep(S) = \{S' \subseteq S : \text{$S'$ is an independent set in $G$}\}$.
\end{definition}

\begin{theorem}\label{olala}
Let $R \supseteq C$ be arbitrary. If there exist positive real numbers $\{\psi_i\}$ such that for every $i \in [m]$,
\begin{align}
\zeta_i := \frac{\gamma_i}{\psi_i} \sum_{S \in \Indep(\Gamma_{R}(i))}  \prod_{j \in S} \psi_j < 1 \enspace , \label{eq:mc3}
\end{align}
then for every permutation $\pi$, the recursive walk reaches a sink within $(T_0+s)/\delta$ steps with probability at least $1-2^{-s}$, where $\delta = 1 - \max_{i \in [m]} \zeta_i > 0$, and
\[
T_0 	 = 
\log_2 \left( \max_{ \sigma \in \Omega } \frac{\theta(\sigma)  }{ \mu(\sigma) }\right) + \log_2 \left( \sum_{S \subseteq \Indep \left( \Span(\theta) \right)  } \prod_{j \in S} \psi_j   \right)
 \enspace  .
\]
\end{theorem}

\begin{remark}
Theorem~\ref{olala} strictly improves Theorem~\ref{asymmetric} since for $R=C$: (i) the summation in~\eqref{eq:mc3} is only over the subsets of $\Gamma_R(i)$ that are independent in $G$, instead of  all subsets of $\Gamma_R(i)$ as in~\eqref{eq:mc2}, and (ii) similarly for $T_0$, the summation is only over the independent subsets of $\Span(\theta)$, rather than all subsets of $\Span(\theta)$.
\end{remark}

\begin{remark}
Theorem~\ref{olala} can be strengthened by introducing for each flaw $f_i \in F$ a permutation $\pi_i$ of $\Gamma_{R}(i)$ and replacing $\pi$ with $\pi_i$  in line~\ref{code:act} the of Recursive Walk. With this change in~\eqref{eq:mc3} it suffices to sum only over $S \subseteq \Gamma_{R}(i)$ satisfying the following: if the subgraph of $R$ induced by $S$ contains an arc $j \to k$, then $\pi_{i}(j) \ge \pi_i(k)$. As such a subgraph can not contain both $j \to k$ and $k\to j$ we see that $S \in \Indep(\Gamma_R(i))$.
\end{remark} 

\subsection{A General Theorem}

Theorems~\ref{asymmetric} and~\ref{olala} are instantiations of a general theorem we develop for analyzing focused stochastic local search algorithms. Before stating the theorem we briefly discuss its derivation in order to motivate its form. Recall that a focused local search algorithm $\mathcal{A}$ amounts to a flaw choice mechanism driving a random walk on a multidigraph $D$ with transition probabilities $\rho$ and starting state distribution $\theta$. 

To bound the probability that $\mathcal{A}$ runs for $t$ or more steps we partition the set of all $t$-trajectories into equivalence classes, bound the total probability of each class, and sum the bounds for the different classes. Specifically, the partition is according to the $t$-sequence of the first $t$ flaws addressed.
\begin{definition}\label{def:w_t}
For any integer $t\ge 1$, let $\mathcal{W}_t(\mathcal{A})$ denote the set containing all $t$-sequences of flaws that have positive probability of being the first $t$ flaws addressed by $\mathcal{A}$.
\end{definition}

In general, the content of $\mathcal{W}_t(\mathcal{A})$ is an extremely complex function of flaw choice. An essential idea of our analysis is to overapproximate it by \emph{syntactic} considerations capturing the following necessary condition for $W \in \mathcal{W}_t(\mathcal{A})$:  while the very first occurrence of any flaw $f_j$ in $W$ may be attributed to $f_j \ni \sigma_1$, every subsequent occurrence of $f_j$ must be preceded by a distinct occurrence of a flaw $f_i$ that  ``assumes responsibility" for $f_j$, e.g., a flaw $f_i$ that potentially causes $f_j$. Definition~\ref{def:traceable} below establishes a framework for bounding $\mathcal{W}_t(\mathcal{A})$ by relating flaw choice with responsibility by (i) requiring that the flaw choice mechanism is such that the elements of $\mathcal{W}_t(\mathcal{A})$ can be unambiguously represented forests with $t$ vertices, while on the other hand (ii) generalizing the subsets of flaws for which a flaw $f_i$ may be responsible from subsets of $\Gamma(i)$ to arbitrary subsets of flaws, thus enabling responsibility shifting.

\begin{definition}\label{def:traceable}
We will say that algorithm $\mathcal{A}$ is \emph{traceable} if there exist sets $\mathrm{Roots}(\theta) \subseteq 2^{[m]}$ and $\mathrm{List}(1) \subseteq 2^{[m]}, \ldots, \mathrm{List}(m) \subseteq 2^{[m]}$ such that for every $t \ge 1$, the flaw sequences in $\mathcal{W}_t(\mathcal{A})$ can be injected into unordered rooted forests with $t$ vertices that have  the following properties:
\begin{enumerate}
\item 
Each vertex of the forest is labeled by an integer $i \in [m]$.\label{cond:label}
\item 
The labels of the roots of the forest are distinct and form an element of $\mathrm{Roots}(\theta)$.\label{cond:roots}
\item 
The indices labeling the children of each vertex are distinct.\label{cond:distinct}
\item
If a vertex is labelled by $i \in [m]$, then the labels of its children form an element of $\mathrm{List}(i)$.
\label{cond:list}
\end{enumerate}
\end{definition}

In~\cite{AIJACM} it was shown that both the simple random walk algorithm in Theorem~\ref{asymmetric} and the recursive walk algorithm in Theorem~\ref{olala} are traceable. Specifically,  the set $\mathcal{W}_t$ of the former can be injected into so-called \emph{Break Forests}, so that Definition~\ref{def:traceable} is satisfied, with $\mathrm{Roots}(\theta) = 2^{\Span(\theta)}$ and $\mathrm{List}(i) = 2^{\Gamma_R(i)}$. For the latter, $\mathcal{W}_t$ can be analogously injected into so-called \emph{Recursive Forests} with $\mathrm{Roots}(\theta) = \Indep(\Span(\theta))$ and $\mathrm{List}(i) = \Indep(\Gamma_R(i))$. Thus, Theorems~\ref{asymmetric},\ref{olala} follow readily from Theorem~\ref{lem:master} below. 
\begin{theorem}[Main result]\label{lem:master}
If algorithm $\mathcal{A}$ is  traceable and there exist positive real numbers $\{\param_i\}_{i \in [m]}$ such that for every  $i \in [m]$,
\begin{align}
\zeta_i := \frac{\gamma_i}{\param_i } \sum_{S \in \List(i)}  \prod_{j \in S} \param_j < 1 \enspace \label{genikoteri} , 
\end{align}
then $\mathcal{A}$ reaches a sink within $(T_0+s)/\delta$ steps with probability at least $1-2^{-s}$, where $\displaystyle{\delta = 1 - \max_{i \in [m]} \zeta_i}$ and
\[
 T_0 		 = \log_2 \left( \max_{ \sigma \in \Omega } \frac{\theta(\sigma)  }{ \mu(\sigma) }\right) + \log_2 \left( \sum_{S \in \mathrm{Roots}(\theta) } \prod_{j \in S} \psi_j   \right)\enspace .
\]
\end{theorem}

Theorem~\ref{lem:master} also implies the  ``LeftHanded Random Walk" result of~\cite{AIJACM} and extends it to non-uniform transition probabilities, since that algorithm is also traceable. Notably, in the LeftHanded LLL introduced by Pedgen~\cite{PegdenLLLL} and which inspired the algorithm, the flaw order $\pi$ can be chosen in a \emph{provably} beneficial way, unlike in the algorithms of Theorems~\ref{asymmetric} and~\ref{olala}, which are indifferent to $\pi$. Establishing this goodness, though, entails attributing responsibility very differently from what is suggested by the causality digraph, making full use of the power afforded by traceability and Theorem~\ref{lem:master}.

\subsection{Resampling Oracles via Atomic Actions}

To get a constructive result by LLL algorithmization via resampling oracles, i.e., given $\Omega, F$, and $\mu$, we must design $(D,\rho)$ that regenerate $\mu$ at every flaw $f_i \in F$. 
This can be a daunting task in general. 
We simplify this task greatly for atomic action digraphs. Such digraphs capture algorithms that appear in several settings, e.g., the Moser-Tardos algorithm when flaws correspond to partial assignments, the algorithm of Harris and Srinivasan for permutations~\cite{SrinivasanPerm}, and others (see~\cite{AIJACM}). While atomicity may seem an artificial condition, it is actually a natural way to promote search space exploration, as it is equivalent to the following: distinct states $\sigma, \sigma' \in f_i$ must have disjoint actions, i.e., $A(i,\sigma) \cap A(i,\sigma') = \emptyset$. In most settings atomicity can be achieved in a straightforward manner. For example, in the variable setting 
%
%
atomicity is implied by an idea that is extremely successful in practice, namely ``focus"~\cite{papafocus, walksat, circumspect}: every state transformation should be the result of selecting a flaw present in the current state and modifying only the variables of that flaw.

Theorem~\ref{atomic_oracles} asserts that when the action digraph $D$ must be atomic, then in order to regenerate $\mu$ at $f_i$ it is sufficient (and necessary) for the states in each set $A(i,\sigma)$ to have total probability given by~\eqref{synthiki_atomic}. Equation~\eqref{eq:harmo_def} then automatically provides appropriate transition probabilities. 
Combined, equations~\eqref{synthiki_atomic}, \eqref{eq:harmo_def} offer strong guidance in designing 
resampling oracles in atomic digraphs.

\begin{theorem} \label{atomic_oracles}
If $D$ is atomic and $(D,\rho)$ regenerate $\mu$ at $f_i$, then for every $\sigma \in f_i$:
\begin{eqnarray}
\sum_{\tau \in A(i,\sigma) } \mu(\tau) 
	& = & \frac{\mu(\sigma)}{\mu(f_i)}\label{synthiki_atomic} \\
\rho_i(\sigma,\tau) 
	& = & \frac{\mu(\tau)}{ \sum_{\sigma' \in A(i,\sigma) } \mu(\sigma')} \quad\text{for every $\tau \in A(i,\sigma)$} \label{eq:harmo_def} \enspace .
\end{eqnarray}
\end{theorem}

\subsection{A Sharp Analysis and the Role of Flaw Choice}\label{sec:fc}

Let $W_t$ be the random variable that equals the  sequence of the first $t$ flaws addressed by the algorithm, or $\bot$ if the algorithm reaches a flawless object in fewer than $t$ steps. 
Recall that $\mathcal{W}_t(\mathcal{A})$ denotes the set of all $t$-sequences of flaws that have positive probability of being the first $t$ flaws addressed by an algorithm $\mathcal{A}$, i.e., the range of $W_t$ except $\bot$. Trivially, the probability that $\mathcal{A}$ takes at least $t$ steps equals 
\[
\sum_{ W \in \mathcal{W}_t( \mathcal{A}) } \Pr[W_t = W]  \enspace .
\] 

\begin{theorem}\label{tight}
For any algorithm $\mathcal{A}$ for which $D$ is atomic and $(D,\rho)$ regenerate $\mu$ at every flaw, for every flaw sequence $W = w_1, w_2, \ldots, w_t \in \mathcal{W}_t(\mathcal{A})$, 
\begin{equation}
\Pr[W_t = W] \in [\alpha, \beta] \cdot \prod_{i = 1}^{t} \mu(w_i)
\enspace ,  \label{panw_katw}
\end{equation}
where $\alpha = \min_{\sigma \in \Omega} \theta(\sigma)/\mu(\sigma)$ and $\beta = \max_{\sigma \in \Omega} \theta(\sigma)/\mu(\sigma)$. 
\end{theorem}

Theorem~\ref{tight} tell us that every algorithm where $(D,\rho)$ form atomic resampling oracles, will converge to a flawless object \emph{if and only if} the sum 
\[
\sum_{ W \in \mathcal{W}_t( \mathcal{A}) } \prod_{i = 1}^{t} \mu(w_i)
\] 
converges to zero as $t$ grows. In other words, the quality of the algorithm  depends \emph{solely} on the  set $\mathcal{W}_t(\mathcal{A})$ which, in turn, is determined by flaw choice (and the initial distribution $\theta$). 

In the work of Moser and Tardos for the variable setting~\cite{MT} and of Harris and Srinivasan for the uniform mesure on permutations~\cite{SrinivasanPerm}, flaw choice can be arbitrary and the whole issue ``is swept under the rug''~\cite{mario_survey}. This can be explained as follows. In these settings, due to the symmetry of $\Omega$, we can afford to overapproximate $\mathcal{W}_t(\mathcal{A})$ in a way that completely ignores flaw choice, i.e., we can deem flaw choice to be adversarial, and still recover the LLL condition. Theorem~\ref{tight} shows that this should not be confused with deeming flaw choice ``irrelevant" for such algorithms. Exactly the opposite is true, a fact also established experimentally~\cite{MTexper}: the Moser-Tardos algorithm, in practice, succeeds on instances far denser  than predicted by the LLL condition.

Kolmogorov~\cite{Kolmofocs} gave a sufficient condition, called \emph{commutativity}, for arbitrary flaw choice. One can think of commutativity as the requirement that there exists a supergraph of the causality graph satisfying a strong symmetry condition (including that all arcs are bidirectional), for which the LLL condition still holds. However, such symmetries can not be expected to hold in general, something reflected in the requirement of traceability in our Theorem~\ref{lem:master}, and in the specificity of the flaw choice mechanisms in our Theorems~\ref{asymmetric} and~\ref{olala}.
More generally, in~\cite{HV}, Harvey and Vondr\'{a}k provided strong evidence that in the absence of commutativity, specific flaw choice is necessary to match Shearer's criterion for the LLL.

\subsection{Comparison with Resampling Oracles}

Harvey and Vondr\'{a}k~\cite{HV} proved that in the setting of resampling oracles, i.e., no distortion, when the causality graph is symmetric, if one resamples a maximal independent set of bad events each time, the resulting algorithm succeeds even under Shearer's condition. (Notably, Shearer's condition, involving an exponential number of terms, is not used in applications). As a corollary, they prove that in this setting, in~\eqref{eq:mc2}, strict inequality ($<$) can be replaced with inequality ($\le$).   As our results are over arbitrary directed causality graphs, for which no analogue to Shearer's condition exists, we do not have an analogous result. However, for the case where the causality graph is symmetric (undirected), the third author showed~\cite{Kolmofocs}  that the analogue of Shearer's lemma holds in our framework. That is, if the conditions that result when in the standard Shearer lemma one replaces probabilities with charges are satisfied, then the algorithm that resamples a maximal independent set of bad events each time succeeds.


\section{Bounding the Probabilities of Trajectories}

To bound the probability that an algorithm $\mathcal{A}$ runs for $t$ or more steps we partition  its $t$-trajectories into equivalence classes, bound the total probability of each class, and sum the bounds for the different classes.  Formally, for a trajectory $\Sigma = \sigma_1 \xrightarrow{w_1}   \sigma_2 \xrightarrow{w_2}  \cdots$ we let $W(\Sigma) = w_1,w_2,\ldots$ denote its \emph{witness} sequence, i.e., the sequence of flaws addressed along $\Sigma$ (note that $\Sigma$ determines $W(\Sigma)$ as flaw choice is deterministic). We let $W_t(\Sigma) = \perp$ if $\Sigma$ has fewer than $t$ steps, otherwise we let $W_t(\Sigma)$ be the $t$-prefix of $W(\Sigma)$. Slightly abusing notation, as mentioned, we let $W_t = W_t(\Sigma)$ be the random variable when $\Sigma$ is the trajectory of the walk, i.e., selected according to $(D,\rho,\theta)$ and the flaw choice mechanism. Finally, recall that $\mathcal{W}_{t} = \mathcal{W}_{t}(\mathcal{A}) $ denotes the range of $W_t$ for algorithm $\mathcal{A}$ except for $\perp$, i.e., $\mathcal{W}_{t}(\mathcal{A})$ is the set of $t$-sequences of flaws that have positive probability of being the first $t$ flaws addressed by $\mathcal{A}$, as per Definition~\ref{def:w_t}. Thus,
\[
\Pr[\text{Algorithm $\mathcal{A}$ takes $t$ or more steps}] = 
\sum_{W \in \mathcal{W}_{t}(\mathcal{A})} \Pr[W_t = W] \enspace .
\]

Key to our analysis will be the derivation of an upper bound for $\Pr[W_t = W]$ that holds for \emph{arbitrary} $t$-sequences of flaws, i.e., not necessarily elements of $\mathcal{W}_{t}(\mathcal{A})$, and which factorizes over the flaws in $W$. For an arbitrary sequence of flaws $A = a_1,\ldots,a_t$, let us denote by $[i]$ the index $j \in [m]$ such that $a_i = f_j$. 
\begin{lemma}\label{lemma:regeneration}
Let $\xi = \xi(\theta,\mu) = \max_{\sigma \in \Omega} \{\theta(\sigma)/\mu(\sigma)\}$. For every sequence of flaws $W = w_1, \ldots, w_t$,
\[
\Pr[W_t = W] \le \xi \prod_{i = 1}^{t} \gamma_{[i]} \enspace .
\]
\end{lemma}

\begin{proof}
We claim that for every $t \geq 0$, every $t$-sequence of flaws $W$, and every state $\tau \in \Omega$,
\begin{equation}\label{eq:goal}
\Pr[W_t = W \cap \sigma_{t+1} = \tau] 
\le \xi   \cdot \prod_{i = 1}^{t} \gamma_{[i]} \cdot \mu(\tau) \enspace .
\end{equation}
Summing~\eqref{eq:goal} over all $\tau \in \Omega$ proves the lemma.

To prove our claim we proceed by induction on $|W|$ after recalling that for every $i \in [m]$ and $\tau \in \Omega$, by the definition of $\{\gamma_i\}_{i \in [m]}$, 
\begin{equation}\label{eq:aro}
\sum_{\sigma \in f_i} \mu(\sigma) \rho_i(\sigma,\tau) \le \gamma_i \cdot \mu(\tau) \enspace . 
\end{equation}

For $|W|=0$ the claim holds because $\Pr[\sigma_1 = \tau] = \theta(\tau) \le \xi \mu(\tau)$ for all $\tau \in \Omega$, by the definition of $\xi$.

Assume that~\eqref{eq:goal} holds for all $s$-sequences of flaws, for some $s \ge 0$. Let $A' = A, f_i$ be any sequence of $s+1$ flaws and let $\tau \in \Omega$ be arbitrary. The first inequality below is due to the fact that since $f_i$ is the last flaw in $A'$ a necessary (but not sufficient) condition for  the event $W_{s+1} = A' $  to occur is that $f_i$ is present in the state that results after the flaws in $A$ have been addressed (it is not sufficient as $\mathcal{A}$ may choose to address a flaw other than $f_i$). The second inequality  follows from the inductive hypothesis, while the third from~\eqref{eq:aro}. 
\begin{eqnarray*}
\Pr[W_{s+1} = A' \cap \sigma_{s+2} = \tau] 
& \le &  
\sum_{\sigma \in f_i}  \rho_i(\sigma,\tau)  \Pr[W_s = A \cap \sigma_{s+1} = \sigma] \\
& \le &  
\xi \cdot \prod_{i = 1}^{s} \gamma_{[i]}  \cdot \sum_{\sigma \in f_i}  \mu(\sigma) \cdot \rho_i (\sigma, \tau) \\
& \le &  
\xi  \cdot \prod_{i = 1}^{s+1} \gamma_{[i]} \cdot  \mu(\tau) \enspace .
\end{eqnarray*}
\end{proof}

\section{Proof of Theorems~\ref{atomic_oracles} and~\ref{tight}}\label{sec:misc_proofs}

We first identify for every digraph--measure pair $(D, \mu)$ certain transition probabilities $\rho$ as special.
\begin{harmonic}
$(D,\rho,\mu)$ are \emph{harmonic} if for every $i \in [m]$ and every transition $(\sigma,\tau ) \in   f_i \times A(i,\sigma)$,
\begin{equation}\label{eq:harmonic_rho_def} 
\rho_i(\sigma,\tau) =  \frac{\mu(\tau)}{ \sum_{\sigma' \in A(i,\sigma) } \mu(\sigma')} \propto \mu(\tau) \enspace .
\end{equation}
\end{harmonic}
In words, when $(D,\rho,\mu)$ are harmonic $\rho_i$ assigns to each state in $A(i,\sigma)$ probability proportional to its probability under $\mu$. It is easy to see that $(D,\rho,\mu)$ are harmonic both in the celebrated algorithm of Moser and Tardos~\cite{MT} for the variable setting and in the algorithm of Harris and Srinivasan~\cite{SrinivasanPerm} for the uniform measure on permutations. What makes harmonic $(D,\rho,\mu)$ combinations special is that for any pair $(D,\mu)$, taking $\rho$ so that $(D,\rho,\mu)$ are harmonic, can be easily seen to minimize the expression 
\[
\max_{\tau \in A(i,\sigma)} \left\{\rho_i(\sigma,\tau) \, \frac{\mu(\sigma)}{\mu(\tau)}\right\} 
\]
for every $\sigma \in f_i$ \emph{simultaneously}. For atomic $D$ this suffices to minimize the charge $\gamma_i$ over all possible $\rho$.

\begin{proof}[Proof of Theorem~\ref{atomic_oracles}]
If $D$ is atomic, $\mu > 0$, and $(D,\rho)$ regenerate $\mu$ at every flaw $f_i$, it follows that for every $\tau \in \Omega$ there is \emph{exactly} one $\sigma \in f_i$ such that $\rho_i(\sigma,\tau)>0$. (And also that $\bigcup_{\sigma \in f_i} A(i,\sigma) = \Omega$). Therefore, regeneration at $f_i$ in this setting is equivalent to:
\begin{equation}\label{eq:ptoma}
\text{For every $\tau \in \Omega$ and the unique $\sigma$ such that $\tau \in A(i,\sigma)$:}\quad
 \rho_i(\sigma, \tau) = \mu(\tau)  \frac{\mu(f_i) }{ \mu(\sigma)}  \enspace .
\end{equation}
(Note that for given $D, \mu$ there may be no $\rho$ satisfying~\eqref{eq:ptoma}, as we also need that $\sum_{\tau \in A(i,\sigma)}\rho_i(\sigma,\tau) = 1$.)

Since $\rho_i(\sigma, \tau) \propto \mu(\tau)$ in~\eqref{eq:ptoma} we get~\eqref{eq:harmo_def}. Summing~\eqref{eq:ptoma} over $\tau \in A(i,\sigma)$ yields~\eqref{synthiki_atomic}.
\end{proof}

\begin{proof}[Proof of Theorem~\ref{tight}]
Lemma~\ref{lemma:regeneration}, valid for any $(D,\rho,\mu,\theta)$, readily yields the upper bound. For the lower bound, we start by recalling that in the proof of Theorem~\ref{atomic_oracles} we showed that if $D$ is atomic and $(D,\rho)$ regenerate $\mu$ at $f_i$, then $\bigcup_{\sigma \in f_i} A(i,\sigma) = \Omega$. Therefore, if $W \in \mathcal{W}_t$, since $(D,\rho)$ regenerate $\mu$ at every $f_i$, for every $\tau \in \Omega$ there exists $\Sigma^{\tau} = \sigma_1^{\tau}, \ldots, \sigma_{t+1}^{\tau}$ such that $W(\Sigma^{\tau}) =  W$ and $\sigma_{t+1}^{\tau} = \tau$. Trivially,
\[
\Pr \left[\Sigma^{\tau} \right] = 
\theta(\sigma^{\tau}_1) \prod_{i = 1}^{t} \rho_{[i]}(\sigma^{\tau}_i, \sigma^{\tau}_{i+1}) \enspace . 
\]

Since $D$ is atomic and $(D,\rho)$ regenerate $\mu$ at every flaw, \eqref{eq:ptoma} applies, yielding
\[
\rho_{[i]}(\sigma^{\tau}_i, \sigma^{\tau}_{i+1}) = 
\mu(w_i) \frac{\mu(\sigma^{\tau}_{i+1})}{\mu(\sigma^{\tau}_{i})} \enspace .
\]
Thus, by telescoping, 
\begin{equation}\label{eq:refined_lb}
\Pr[\Sigma^{\tau}] = 
\theta(\sigma^{\tau}_1) \prod_{i = 1}^{t} \mu(w_i) \frac{\mu(\sigma^{\tau}_{i+1})}{\mu(\sigma^{\tau}_{i})} =
\frac{\theta(\sigma^{\tau}_1)}{\mu(\sigma^{\tau}_{1})}   \mu(\tau)
\prod_{i = 1 }^{t} \mu(w_i)   \enspace . 
\end{equation}
Summing~\eqref{eq:refined_lb} over $\tau \in \Omega$ gives the lower bound
\[
\Pr[W_t = W] \ge \min_{\sigma \in \Omega} \frac{\theta(\sigma)}{\mu(\sigma)} \prod_{i = 1 }^{t} \mu(w_i)    \enspace .
\]
\end{proof}

\section{ Proof of Theorem~\ref{lem:master}}\label{sec:bpproof}

Per the hypothesis of Theorem~\ref{lem:master}, the sequences in $\mathcal{W}_t$ can be injected into a set of rooted forests with $t$ vertices that satisfy the properties of Definition~\ref{def:traceable}. Let $\widetilde {\mathcal{W}_t} \supseteq  \mathcal{W}_t$ be the set of \emph{all} forests with $t$ vertices  that satisfy the properties of Definition~\ref{def:traceable}. By Lemma~\ref{lemma:regeneration}, to prove the theorem it suffices to prove that  $\max_{ \sigma \in \Omega} \frac{\theta(\sigma) }{\mu(\sigma)}   \sum_{W \in \widetilde{ \mathcal{W}}_t} \prod_{i = 1}^{t} \gamma_{[i]} $ is exponentially small in $s$ for $t=T_0+s$.




 To proceed, we use ideas  from~\cite{PegdenIndepen}. Specifically, we introduce a branching process that produces only forests in   $\widetilde { \mathcal{W}_t }$  and bound $\sum_{W \in \widetilde{\mathcal{W}_t}} \prod_{i = 1}^{t} \gamma_{[i]} $ by analyzing it.  Given any real numbers $0 < \param_i < \infty$ we define $x_i = \frac{\param_i}{\param_i +1} $ and write $\mathrm{Roots}(\theta) = \mathrm{Roots}$ to simplify notation. Recall that neither the trees in each forest, nor the nodes inside each tree are ordered. To start the process we produce the roots of the labeled forest by rejection sampling as follows: For each  $i \in [m]$ independently, with probability $x_i$ we add a root with label $i$. If the resulting set of roots is in $\mathrm{Roots}$ we accept the birth. If not, we delete the roots created and try again. In each subsequent round we follow a very similar procedure. Specifically, at each step, each node $u$ with label $\ell$ ``gives birth", again, by rejection sampling: For each integer $i \in [m]$, independently, with probability $x_{i}$ we add a vertex with label $i$ as a child of $u$. If the resulting set of children of $u$ is in $\List(\ell)$ we accept the birth. If not, we delete the children created and try again. It is not hard to see that this process creates every  forest in $\widetilde { \mathcal{W}_t }$  with positive probability. Specifically, for a vertex labeled by $\ell$, every set $S \not\in\List(\ell) $ receives probability 0, while every set $S\in \List(\ell)$ receives probability proportional to
\[
w(S) = \prod_{g \in S} x_g \prod_{h \in  [m] \setminus S} \left(1- x_h\right) \enspace .
\]

To express the exact probability received by each $S\in\List(\ell)$ we define
\begin{equation}\label{eq:d_def}
Q(S)  :=   \frac{ \prod_{g \in S} x_g }{\prod_{g \in S}(1 - x_g)   }  = \prod_{g\in S} \param_g 
\end{equation}
and let $Z = \prod_{i \in   [m]  }\left(1 - x_i\right) $.
We claim that $w(S) = Q(S) \, Z $. To see the claim observe that
\[
\frac{ w(S)}
{Z}
=  \frac{ \prod_{g \in S} x_g \prod_{h \in  [m]  \setminus S} \left(1- x_h\right) }
{\prod_{i \in   [m]  }\left(1 - x_i\right) }  
=  \frac{ \prod_{g \in S} x_g }{\prod_{g \in S}(1 - x_g)   }  
= Q(S) \enspace .
\]
Therefore, each $S\in\List(\ell)$ receives probability equal to
\begin{equation}\label{eq:sing_birth}
\frac{w(S)}{\sum_{B \in \List(\ell)} w(B)}
=
\frac
{Q(S) Z}
{\sum_{B \in\List(\ell)} Q(B) Z}
=\frac{Q(S)}{\sum_{B \in\List(\ell)} Q(B)}
 \enspace .
\end{equation}

Similarly, each set $ R \in \mathrm{Roots}$ receives probability equal to $Q(R) \left( \sum_{B \in \mathrm{Roots} }Q(B)  \right)^{-1}$.

For each forest $\phi \in \widetilde{\mathcal{W}}_t$ and each node $v$ of $\phi$, let $N(v)$ denote the set of labels of its children and let $\List(v) = \List(\ell) $, where $\ell$ is the label of $v$.

\begin{lemma}\label{branchingLemma}
The branching process described above produces every forest $\phi \in \widetilde{\mathcal{W}}_t$ with probability 
\begin{align*}
p_{\phi} =  \left( \sum_{ S \in  \mathrm{Roots}}
\prod_{i \in S} \param_i  
 \right)^{-1}  \prod_{v \in \phi}  \frac{\param_v}{ \sum_{S \in\List(v)} Q(S)}
\end{align*}
\end{lemma}

\begin{proof} 
Let $R$ denote the roots of $\phi$. By~\eqref{eq:sing_birth},
\begin{align*}
p_{\phi} & =  \frac{ Q(R) }{ \sum_{S \in \mathrm{Roots}} Q(S)  }  \prod_{v \in \phi} \frac{Q(N(v))}{\sum_{S \in\List(v)} Q(S)} \\
& =  \frac{ Q(R) }{ \sum_{S \in \mathrm{Roots}} Q(S)  } \cdot \frac{\prod_{v \in \phi \setminus R} \param_v} {\prod_{v \in \phi} \sum_{S \in\List(v)} Q(S)}\\
& = \left( \sum_{S \in \mathrm{Roots}} Q(S) \right)^{-1}   \prod_{v \in \phi} \frac{\param_v }{\sum_{S \in\List(v)} Q(S)}  \enspace .
\end{align*}
\end{proof}

Notice now that
\begin{eqnarray}
\sum_{W \in \widetilde{\mathcal{W}_t } } \prod_{i =1}^{ t} \gamma_{[i] } & = & \sum_{W \in \widetilde{\mathcal{W}_t }} \prod_{i =1}^{ t}   \frac{\zeta_{[i] }\, \param_{[i] }}{ \sum_{S \in \List({[i] })} Q(S) }  \nonumber \\
& \le & \left( \max_{ i \in [m]} \zeta_i  \right)^{t} \sum_{W \in \widetilde{\mathcal{W}_t }}  \prod_{i =1}^{ t}   \frac{\param_{[i] }}{ \sum_{S \in \List([i])} Q(S)}  \nonumber \\
& = & \left( \max_{ i \in [m]} \zeta_i  \right)^{t}  \sum_{W \in \widetilde{\mathcal{W}_t } }\left( p_{W}  \sum_{S \in \mathrm{Roots} }  Q(S) \right)   \nonumber \\
& = & \left( \max_{ i \in [m]} \zeta_i  \right)^{t}  \sum_{ S \in \mathrm{Roots} } Q(S) \label{no_kagkouro_versions} \enspace .
\end{eqnarray}

Using~\eqref{no_kagkouro_versions} we see that the binary logarithm of the probability that the walk does not encounter a flawless state within $t$ steps is at most $t  \log_2 \left(  \max_{i \in F} \zeta_i \right) + T_0$, where 
\begin{eqnarray*}
 T_0 		& =& \log_2 \left( \max_{ \sigma \in \Omega } \frac{\theta(\sigma)  }{ \mu(\sigma) }\right) + \log_2 \left( \sum_{S \in \mathrm{Roots} } \prod_{i \in S} \param_i   \right)\enspace .
\end{eqnarray*}

Therefore, if $t = (T_0 + s) / \log_2 (1/ \max_{i \in F} \zeta_i) \le (T_0 + s) / \delta$, the probability that the  random walk  on $D$ does not reach a flawless state within $t$ steps is at most  $2^{-s}$.

\section{ Application to Acyclic Edge Coloring }\label{AECARA}

\subsection{Earlier Works and Statement of Result}

An edge-coloring of a graph is \emph{proper} if all edges incident to each vertex have distinct colors. A proper edge coloring is \emph{acyclic} if it has no bichromatic cycles, i.e., no cycle receives exactly two (alternating) colors. Acyclic Edge Coloring (AEC), was originally motivated by the work of Coleman et al.~\cite{coleman1,coleman2} on the efficient computation of Hessians. The smallest number of colors, $\chi'_a(G)$, for which a graph $G$ has an acyclic edge-coloring can also be used to bound other parameters, such as the oriented chromatic number~\cite{orient_col} and the star chromatic number~\cite{star_col}, both of which have many practical applications. The first general linear upper bound for $\chi'_a$ was given by Alon et al.~\cite{NogaLLL} who proved $\chi'_a(G) \le 64 \Delta(G)$, where $\Delta(G)$ denotes the maximum degree of $G$. This bound was improved to $16\Delta$ by Molloy and Reed~\cite{MRAEC} and then to $9.62(\Delta-1)$ by Ndreca et al.~\cite{Ndreca}. Attention to the problem was recently renewed due to the work of Esperet and Parreau~\cite{acyclic} who proved $\chi'_a(G) \le  4(\Delta-1)$, via an entropy compression argument, a technique that goes beyond what the LLL can give for the problem. Very recently, Giotis et al.~\cite{kirousis} improved the result of~\cite{acyclic} to $3.74 \Delta$.\smallskip 

We give a bound of $2\Delta+o(\Delta)$ for (simple) graphs of bounded degeneracy. Recall that a graph $G$ is $d$-degenerate if its vertices can be ordered so that every vertex has at most $d$ neighbors greater than itself. Thus, we not only cover a significant class of graphs, but demonstrate that our method can incorporate global graph properties. For example, if $\mathcal{G}_d$ denotes the set of all $d$-degenerate graphs, then all planar graphs are in $\mathcal{G}_5$, while all graphs with treewidth or pathwidth at most $d$ are in $\mathcal{G}_d$
(for more on degenerate graphs see~\cite{jensen}). We prove the following.
\begin{theorem}\label{Aecaki}
Every $d$-degenerate graph of maximum degree $\Delta$ has an acyclic edge coloring with $\lceil (2+\epsilon)\Delta\rceil$ colors than can be found in polynomial time, where $\epsilon = 4\sqrt{d/\Delta}$.
\end{theorem}

\subsection{Background}

As will become clear shortly, the main difficulty in AEC comes from the short cycles of $G$, with 4-cycles being the toughest. This motivates the following definition.
\begin{definition} 
Given a graph $G=(V,E)$ and a, perhaps partial, edge-coloring of $G$, say that color $c$ is \emph{4-forbidden for $e \in E$} if assigning $c$ to $e$ would result in either a violation of proper-edge-coloration, or in a bichromatic 4-cycle containing $e$. Say that $c$ is 4-available if it is not 4-forbidden.
\end{definition}
Similarly to~\cite{acyclic,kirousis} we use the following observation that the authors of~\cite{acyclic} attribute to Jakub Kozik. 
\begin{lemma}[\cite{acyclic}]\label{lem:2D}
In any proper edge-coloring of $G$ at most $2(\Delta-1)$ colors are 4-forbidden for any $e \in E$.
\end{lemma}
\begin{proof}
The 4-forbidden colors for $e = \{u,v\}$ can be enumerated as: (i) the colors on edges adjacent to $u$, and (ii) for each edge $e_v$ adjacent to $v$, either the color of $e_v$ (if no edge with that color is adjacent to $u$), or the color of some edge $e'$ which together with $e, e_v$ and an edge adjacent to $u$ form a cycle of length $4$. 
\end{proof}

Armed with Lemma~\ref{lem:2D}, the general idea is to use a palette $P$ of size $2(\Delta -1) + Q$ colors so that whenever we (re)color an edge $e$ there will be at least $Q$ colors 4-available for $e$ (of course, coloring $e$ may create one or more bichromatic cycles of length at least 6).  At a high level, similarly to~\cite{kirousis}, our algorithm will be:
\begin{itemize}
\item
Start at a proper edge-coloring with no bichromatic 4-cycles. 
\item
While bichromatic cycles of length at least 6 exist, recolor the edges of one with 4-available colors.
\end{itemize}
Note that to find bichromatic cycles in a properly edge-colored graph we can just consider each of the $\binom{|P|}{2}$
pairs of distinct colors from $P$ and seek cycles in the subgraph of the correspondingly colored edges.

\subsection{Applying our Framework}

Given $G=(V,E)$ and a palette $P$ of $2(\Delta -1) + Q$ colors, let $\Omega$ be the set of all proper edge-colorings of $G$ with no monochromatic 4-cycle. Fix an arbitrary ordering $\pi$ of $E$ and an arbitrary ordering $\chi$ of $P$. For every even cycle $C$ of length at least 6 in $G$ fix (arbitrarily) two adjacent edges $e_1^C, e_2^C$ of $C$. \smallskip 

-- Our distribution of initial state $\theta$ assigns all its probability mass to the following $\sigma_1 \in \Omega$: color the edges of $E$ in $\pi$-order, assigning to each edge $e \in E$ the $\chi$-greatest 4-available color.

-- For every even cycle $C$ of length at least 6 we define the flaw $f_C = \{\sigma \in \Omega: C \text{ is bichromatic}\}$. Thus, a flawless $\sigma \in \Omega$ is an acyclic edge coloring of $G$.

-- The set of actions for addressing $f_C$ in state $\sigma$, i.e., $A(C, \sigma)$, comprises all $\tau \in \Omega$ that may result from the following procedure: uncolor all edges of  $C$ except for $e_1^{C},e_2^{C}$; go around $C$, starting with the uncolored edge that is adjacent to $e_2^C$, etc., assigning to each uncolored edge $e \in C$ one of the  4-available colors for $e$ at the time $e$ is considered. Thus, by lemma~\ref{lem:2D}, $|A(C,\sigma)| \ge Q^{|C|-2}$.
 
\begin{lemma}\label{atomicityclaim}
For every flaw $f_C$ and  state $\tau \in \Omega$, there  is  at most $1$ arc $\sigma \xrightarrow{C}  \tau$.
\end{lemma}
\begin{proof}
Given $\tau$ and $C$, to recover the previous state $\sigma$ it suffices to extend the bicoloring in $\tau$ of $e_1^C, e_2^C$ to the rest of $C$ (since $C$ was bichromatic in $\sigma$ and only edges in $C\setminus\{ e_1^{C},e_2^{C}  \}$ were recolored).
\end{proof}
Thus, taking $\mu$ to be uniform and $\rho$ such that $(D,\rho,\mu)$ is harmonic yields $\gamma_C \le Q^{-|C|+2}$. \medskip

Let $R$ be the symmetric directed graph with one vertex per flaw where $f_C \rightleftarrows f_{C'}$ iff $C \cap C' \ne \emptyset$. Since a necessary condition for $f_C$ to potentially cause $f_{C'}$ is that $C \cap C' \ne \emptyset$, we see that $R$ is a supergraph of the causality digraph. Thus, if we run the {\sc Recursive Walk} algorithm with input $R$, to apply Theorem~\ref{olala} we need to evaluate for each flaw $f_C$ a sum over the subsets of $\Gamma_R(C)$ that are independent in $R$.

For $n\ge 2$, let $g(n) = \max_{e\in E} |\{(2n+2)\mbox{-cycles in $G$ that contain $e$}\}|$. 
Let $\alpha_2,\alpha_3,\ldots$ be positive numbers such that $\beta=\sum_{n=2}^\infty \alpha_n<\infty$.
We will use $\psi_C=\psi(n)=\alpha_n/g(n)$ for a cycle $C$ of length $2n+2$.
For a set of edges $X=\{e_1,\ldots,e_k\}$, let ${\tt Ind}_X$ denote the set of all $k$-sets of cycles $S=\{C_1,\ldots,C_k\}$, 
where $e_i\in C_i$ for every $i$, and where the cycles are edge-disjoint and, therefore, independent in $R$. Then,
$$
\sum_{S\in{\tt Ind}_X}\prod_{C\in S}\psi_{C} \le    \left( \sum_{n = 2}^{ \infty } g(n) \psi(n)   \right)^{|X|} = \beta^{|X|} \enspace .
$$
Therefore, for each $(2n+2)$-cycle $C$ we can bound~\eqref{eq:mc3} as
\begin{equation}\label{eq:GHASKFAGB}
\frac{\gamma_C}{\psi_C} \sum_{S\in{\tt Ind}(\Gamma_R(C))}\prod_{C'\in S}\psi_{C'} 
\le    \frac{\gamma_C}{\psi_C} \sum_{X\subseteq C}\sum_{S\in {\tt Ind}_X}\prod_{C'\in S}\psi_{C'} 
\le \frac{\gamma_C}{\psi_C} \sum_{X\subseteq C} \beta^{|X|} 
= \frac{\gamma_C}{\psi_C} (1+\beta)^{|C|} \enspace .
\end{equation}

To bound~\eqref{eq:GHASKFAGB} we observe that for arbitrary graphs, trivially,  $g(n) \le (\Delta-1)^{2n}$. Setting $a_n = \lambda^{2n}$  and $Q = \kappa (\Delta-1) $, where $\kappa>0$ and $0 < \lambda <1$ will be specified shortly, we get
\begin{align}\label{tesseradekaksi}
\frac{  \gamma_{ C} }{ \psi_{C} }( 1 + \beta)^{ |C| }  \le \left( \frac{ \Delta-1}{ \lambda Q } \right)^{|C|-2 }    \left(1 + \frac{\lambda^4}{1-\lambda^2}\right)^{ |C| } = \left( \frac{ 1}{ \lambda \kappa } \right)^{|C|-2 }  \left(1 + \frac{\lambda^4}{1-\lambda^2}\right)^{ |C| }\enspace.
\end{align}
Choosing $(\kappa, \lambda) = (2.182,  0.569)$, the right hand side of~\eqref{tesseradekaksi}   becomes strictly less than 1 for every $|C| \ge 6$.  Thus, for general graphs, $4.182(\Delta-1)$ colors suffice to find an acyclic edge coloring efficiently.

\subsection{Graphs of Bounded Degeneracy}

We will prove the following structural lemma relating degeneracy to $g$.

\begin{lemma}
If $G\in{\cal G}_d$ has maximum degree $\Delta$, then $g(n)\le 2 \binom{2n}{n} (d\Delta)^n$.
\label{kral}
\end{lemma}

We let $\alpha_n=\alpha \binom{2n}{n} \lambda^n$ for some $\alpha>0$ and $\lambda\in(0,\frac{1}{4})$, to be specified. Thus, $\beta=\alpha\left(\frac{1}{\sqrt{1-4\lambda}}-1-2\lambda\right)$, since $\sum_{n=0}^\infty\alpha_n=\alpha\sum_{n=0}^\infty\binom{2n}n\lambda^n=\frac{\alpha}{\sqrt{1-4\lambda}}$ (see~\cite{binomsum}). Since $\frac{\gamma_C}{\psi_C}\le\frac{g(n)}{\alpha_n Q^{2n}}\le \frac{2(d\Delta)^n}{\alpha\lambda^n Q^{2n}}$ for any $(2n+2)$-cycle $C$, we see that~\eqref{eq:GHASKFAGB} will be less than $1$ if for every $n \ge 2$,
\[
Q > \left(\frac{2(1+\beta)^2}{\alpha}\right)^{1/2n} \cdot \frac{1+\beta}{\sqrt{\lambda}} \cdot \sqrt{d\Delta} \enspace .
\]
If we take $\alpha, \lambda$ such that $2(1+\beta)^2< (1-\delta)\alpha$, then taking $Q> \frac{1+\beta}{\sqrt{\lambda}} \cdot \sqrt{d\Delta}$ satisfies~\eqref{eq:mc3}. In particular, taking $(\alpha,\lambda)=(2.76, 0.086)$ works, in which case $\frac{1+\beta}{\sqrt{\lambda}}<4$.
Regarding the running time,  it can easily be seen that that $T_0$ is a polynomial in $|E|$, $\Delta$ and the number of colors used (since $\delta$ is a  constant).

\begin{proof}[Proof of Lemma~\ref{kral}]
Fix any edge $e=\{u,v\} \in E$. To enumerate the $(2n+2)$-cycles containing $e$ we will partition them into equivalence classes as follows. First we orient all edges of $G$ arbitrarily to get a digraph $D$. Consider now the two possible traversals of the path $C \setminus \{u,v\}$, i.e., the one starting at $u$ and the one starting at $v$. For each traversal generate a string in $\{0,1\}^{2n}$ whose characters correspond to successive vertices of the path, other than the endpoints, and denote whether the corresponding vertex was entered along an edge oriented in agreement (1) or in disagreement (0) with the direction of travel. Observe that each of the $2n-1$ edges of $C$ that have no vertex from $\{u,v\}$ will create a 1 in one string and a 0 in the other. Therefore, at least one of the two strings will have at least $\lceil (2n-1)/2\rceil = n$ ones. Select that string, breaking ties in favor of the string corresponding to starting at $u$. Then, in the selected string, convert as many of the leftmost 1s as needed to 0s, so that the resulting string has exactly $n$ ones. Finally, prepend a single bit to indicate whether the winning traversal started at $u$ or to $v$. The resulting string is the representative of $C$'s equivalence class. Clearly, there are at most $2\binom{2n}n$ equivalence classes.

To bound the number of cycles in a class we enumerate the possibilities for the $2n$ vertices other that $u,v$, as follows. After reading the bit indicating whether the cycles in the class start at $u$ or at $v$, we interpret each successive character of the representative string to indicate whether we can choose among the out-neighbors or among all neighbors of the current vertex. By the string's construction, there will be exactly $n$ choices of each kind and, therefore, the total number of choices will be at most $\mathrm{Out}^n \Delta^n$, where  $\mathrm{Out}$ is an upper bound on the out-degree of $D$.

To conclude the argument we note that since $G \in \mathcal{G}_d$ we can direct its edges so that every vertex has out-degree at most $d$ by repeatedly removing any vertex $v$ of current degree at most $d$ (it always exists) and, at the time of removal, orienting its current neighbors away from $v$.
\end{proof}

\section*{Acknowledgements}
We are grateful to Dan Kral for providing us with Lemma~\ref{kral} and to Louis Esperet for pointing out an error in our application of Theorem~\ref{olala} to yield Theorem~\ref{Aecaki} in a previous version of the paper. FI is thankful to Alistair Sinclair for many fruitful conversations.

\bibliographystyle{plain}
\bibliography{kolmo}

%

\end{document}